\documentclass[sigconf, nonacm]{acmart}
\makeatletter

\newcommand{\ACMemailhref}[2]{%
  \mbox{\href{#1}{#2}}%
}

\patchcmd{\@mkauthors@iii}
  {\href}
  {\ACMemailhref}
  {}{}

\patchcmd{\@mkauthors@iii}
  {\href}
  {\ACMemailhref}
  {}{}

\makeatother

\usepackage{amsmath}
\usepackage{mathtools}
\usepackage{amsthm}
\usepackage{enumitem}
\usepackage{dsfont}
\usepackage{array}
\usepackage{subcaption}

\usepackage[capitalize,noabbrev]{cleveref}

\theoremstyle{plain}
\newtheorem{theorem}{Theorem}[section]

\newtheorem{lemma}[theorem]{Lemma}

\theoremstyle{definition}

\theoremstyle{remark}

\begin{document}
\newcommand{\defeq}{\vcentcolon=}  
\newcommand{\reals}{{\mbox{\bf R}}}
\newcommand{\ie}{{\it i.e.}}
\newcommand{\Rn}{\mathcal{R}_{+}^n}
\newcommand{\D}{\Delta}

\newcommand{\la}{\lambda}
\newcommand{\de}{\delta}
\newcommand{\ka}{\kappa}

\newcommand{\al}{\alpha}

\newcommand{\tka}{\tilde{\kappa}}
\newcommand{\tKa}{\tilde{\Kappa}}

\newcommand{\CPM}{Constant Product Market}

\newcommand{\IL}{\text{IL}}

\newcommand{\fee}{\tau}
\newcommand{\pen}{\phi}
\newcommand{\win}{\rho}
\newcommand{\imp}{r}

\newcommand{\E}{\mathbb{E}}
\newcommand{\dd}{\,d}

\newcommand{\saife}{\texttt{SAiFE\_gym}\xspace}

\newcommand{\counts}{\mathcal{N}}

\NewDocumentCommand{\z}{e{^}}{%
  Z\IfValueT{#1}{^{(#1)}}%
}

\newcommand{\clmm}{CPM with CL\xspace}
\newcommand{\clmms}{CPMs with CL\xspace}

\title{%The Role of 
%Risk in 
Concentrated Liquidity Provision: a Reinforcement Learning Perspective}

% \title{Concentrated Liquidity Market Making through Reinforcement Learning}

% \author{Georgios Chionas}
% \email{g.chionas@liverpool.ac.uk}
% \correspondingauthor
% \affiliation{%

\author{Georgios Chionas}
\affiliation{%
    % \institution{School of Computer Science and Informatics \\ University of Liverpool}
    \institution{University of Liverpool}
    \city{Liverpool}
    \country{UK}}
\authornote{Authors contributed equally to the paper.}
\email{g.chionas@liverpool.ac.uk}

\author{Charalampos Kleitsikas}
\affiliation{%
  % \institution{Department of Informatics \\ King's College London}
  \institution{King's College London}
  \city{London}
  \country{UK}}
\authornotemark[1]
\email{charalampos.kleitsikas@kcl.ac.uk}

\author{Stefanos Leonardos}
\affiliation{%
  % \institution{Department of Informatics \\ King's College London}
  \institution{King's College London}
  \city{London}
  \country{UK}
}
\email{stefanos.leonardos@kcl.ac.uk}

\author{Leandro S\'anchez-Betancourt}
\affiliation{%
%  \institution{Mathematical Institute \& Oxford-Man Institute of
% Quantitative Finance \\ University of Oxford}
\institution{University of Oxford}
 \city{Oxford}
 \country{UK}}
\email{leandro.sanchezbetancourt@maths.ox.ac.uk}

\author{Carmine Ventre}
\affiliation{%
   \institution{King's College London}
  \city{London}
  \country{UK}
}
\email{carmine.ventre@kcl.ac.uk}

\renewcommand{\shortauthors}{Chionas et al.}

\begin{abstract}
Automated market makers (AMMs) are a cornerstone of decentralised finance (DeFi).
Constant product markets with concentrated liquidity, such as UniswapV3, are now a well-established design. In these markets, liquidity providers (LPs) face a sequential decision problem: they must decide when to rebalance their positions and which price ranges to allocate capital to as market conditions evolve. We formulate dynamic liquidity provision as a stochastic impulse control problem and use reinforcement learning (RL) to solve it, focusing on providing interpretable solutions. We show that learned policies exhibit rich state-dependent behaviour, allocating liquidity according to mispricing, rebalancing costs, uncertainty, inventory exposure, and heterogeneous risk preferences. These behaviours help compress the left tail of the Profit and Loss (PnL) distribution and avoid catastrophic outcomes under high uncertainty. 
%
%Within our framework, we also illustrate the liquidity concentration observed around prices in real DeFi pools as the superposition of learned policies with heterogeneous risk profiles. 
Finally, we benchmark the RL agents against baseline and sophisticated agents from the AMM microstructure literature and analyse their performance.
\end{abstract}
\maketitle
\section{Introduction}
Decentralised finance (DeFi) has introduced financial exchanges in which trading is executed by smart contracts. The most prominent examples are automated market makers (AMMs), popularised by Uniswap. Most AMMs are implemented as constant product markets (CPMs), a special class of constant function market makers \cite{Angeris20CFMMs}, where prices are determined algorithmically by an invariant function and the reserves locked in the liquidity pool. In these markets, liquidity providers (LPs) deposit capital into a pool and earn fees from users trading against it. While LPs were largely passive in designs such as UniswapV2, the introduction of concentrated liquidity (CL) in UniswapV3 significantly expanded their action space, allowing them to allocate capital to specific price ranges \cite{Adams21v3}.

CL changed the economics of liquidity provision. By concentrating capital around the current price, LPs can earn more trading fees, but narrower positions stop earning fees once the price exits the chosen interval. Hence, LPs face a trade-off between wide positions, which cover more possible prices but earn lower fees, and narrow positions, which are more profitable while active but expose the LP to greater concentration risk. As prices evolve, LPs can dynamically move their capital and create new positions, but doing so incurs blockchain frictions, such as gas fees, and may also require additional trades due to changes in the token composition of the new position. These costs must therefore be weighed against the expected benefits of rebalancing.
Thus, we formulate optimal liquidity provision as a stochastic impulse control problem, which includes stochastic order arrivals, price dynamics, and gas costs. Since Hamilton-Jacobi-Bellman-Quasi-Variational Inequality (HJBQVI)-style methods are intractable in such realistic high-dimensional settings, %We therefore 
we employ reinforcement learning (RL) %techniques 
to study dynamic liquidity provision strategies. Our objective is to provide a tractable computational representation of state-dependent policies in settings where closed-form impulse-control solutions are unavailable.

\vspace{5pt}
\noindent 
%\section
\textbf{Our contribution.} We make the following contributions.
\begin{itemize}[leftmargin=*, nosep]

\item We show that learned policies generate rich and interpretable liquidity provision strategies: agents adapt their rebalancing frequency to blockchain frictions such as gas costs, while choosing position width and asymmetry in response to their risk profile, inventory holdings, existing deployed position, and the degree and direction of arbitrage pressure in the market.

\begin{comment}
\item We show that the superposition of learned policies with heterogeneous risk profiles generates a bell-shaped liquidity distribution around the pool price, providing a model-based explanation for empirical patterns observed in real DeFi pools.
\end{comment}

\item We demonstrate a trade-off between narrow and flexible learned strategies. Narrow-position agents achieve higher fee income in risk-neutral settings but incur higher Impermanent Loss (IL) and gas costs due to more frequent rebalancing, whereas wider-position agents require less active risk management when inventory risk is considered.
\end{itemize}

\vspace{5pt}
\noindent 
%\section
\textbf{Related Work:} 
%\label{sec:related-work}
%
Academic work on AMMs dates back to \cite{Othman2010automated}, but AMMs saw widespread adoption more recently as a computationally cheaper alternative to limit order books in blockchain-based exchanges. Their use in decentralised exchanges was first proposed in \cite{buterin2016prediction,lu2017gnosisdex}, the latter introducing a CPM market maker. This mechanism was first analysed by \cite{angeris2021uniswap}, who identified conditions under which such markets can closely track the price of an external reference market. Subsequently, \cite{Angeris20CFMMs} studied the broader class of constant function market makers (CFMMs) and showed that they can incentivise participants to reveal external market prices. %Despite their dominance, s
Simple AMMs such as UniswapV2 suffer from capital inefficiency. Newer protocols such as UniswapV3 \cite{Adams21v3} address this issue through CL, which allows LPs to specify the price ranges over which their liquidity is active and gives them more granular control over fee generation.
Our work contributes to the work on %emerging area of exploring 
profitable strategies for LPs in CPMs with CL. In fact, empirical studies such as \cite{loesch2021impermanentlossuniswapv3, Di25deviations, milionis2024LVR} show that, currently, % even in the presence of fees, 
LPs are on average incurring a loss. \cite{Heimbach22risksv3} study the risk-return trade-off faced by LPs in UniswapV3, while \cite{Fan22differential} analyse optimal static liquidity provision and show how contract design choices, affect LP profits and trader gas costs. More recently, \cite{Cartea2023predictable} decompose LP losses in concentrated-liquidity CPMs into convexity costs, arising from the convexity of the pool invariant, and opportunity costs, arising from locking assets in the pool instead of investing them in a risk-free account. 
For dynamic liquidity provision, \cite{Cartea2024defi} derive a closed-form solution for frictionless LP allocations, abstracting gas costs and latency, \cite{Fan23StrategicLP} use approximation techniques to obtain profitable strategies, while \cite{Cartea23AMMDesigns} introduce a broader family of AMMs, in which LPs participate in price discovery. RL works for CL such as \cite{xu2025RL, zhang2023RL} employ training in historical data, while \cite{Arcifa2025RL} uses a hybrid setup. Contributing to these efforts, our work focuses primarily on a model-based approach that replicates the full CL mechanics instead, allowing for state-dependent and interpretable analyses in multiple economic regimes.

\section{Preliminaries}
\label{sec:prelims}
%%%%%%%%%%%%%%%%%%%%%%%%%%%%%%%%%%%%%%%%%%%%%%%%%%%%%%

\textbf{Constant Product Markets.}
Let $X$ be a numéraire asset (e.g. USDC) and $Y$ be a risky asset (e.g. ETH). 
A liquidity pool for the pair of assets $X,Y$ defines the marginal exchange rate of asset $Y$ in units of asset $X$, denoted by $Z$. 
The state of a CPM is described by the invariant function $f(q^X,q^Y) = q^X q^Y =\kappa^2$, where $q^X$ and $q^Y$ are the reserves of assets $X$ and $Y$ and $\kappa$ is the liquidity parameter (depth) of the pool, a measure of its available liquidity. The invariant function $f$ connects the state of the CPM before and after a trade is executed. The marginal exchange rate, $Z$, in a CPM is defined as $Z:=q^X/q^Y$. In other words, the marginal exchange rate $Z$ is the price quoted by the CPM at its current state for an infinitesimal trade.
Moreover, the CPM charges a fee, $\fee$, proportional to the size of the trade and thus, $1- \fee$ is the net percentage of the trade after the trading fee.

% Consider now a \emph{Liquidity Taker} (LT) trade of size $y$. Let $\tilde{Z}$ be the average execution trade, that is, $\tilde{Z}$ is the output tokens the LT receives by executing a trade of size $y$. The average execution trade $\tilde{Z}$ depends on the trade size $y$, the reserves $q^X, q^Y$ and the charged fee $\fee$.\footnote{
% Note that, for any positive trade size $y$,  the average execution rate $\tilde{Z}(y)$ is always worse than the marginal exchange rate $Z$.}
% In particular, when the LT buys a quantity $y$ of asset $Y$, she pays an amount of $x = y \cdot \tilde{Z}(y) $, such that 
% %\begin{equation}
% %
% %\label{eq:CPM-buy}
%     $f(q^X + (1-\fee) x, q^Y-y) = \kappa^2 \implies  \tilde{Z}(y) =  \frac{1}{1-\fee} \left(\frac{q^X}{q^Y-y} \right)$.
% %\end{equation}
% Similarly, when the LT sells a quantity $y$ of asset $Y$, she receives $x=y \cdot \tilde{Z}(-y)$ of asset $X$, such that
% %\begin{equation}
% %
% %\label{eq:CPM-sell}
%     $f(q^X - x, q^Y+ (1-\fee) y) = \kappa^2 \implies  \tilde{Z}(-y) = \frac{(1-\fee) q^X}{q^Y+ (1-\fee) y}$.
% %\end{equation}
Consider now a liquidity taker's (LT) trade of size $\Delta x$ (i.e., the LT provides $\Delta x$ units of asset $X$ in exchange for asset $Y$). The LT receives $\Delta y'$ units of asset $Y$ such that $f(q^X + (1-\fee) \Delta x, q^Y-\Delta y') = f(q^X, q^Y)$, which, for a CPM, implies that
$ \Delta y' = q^Y - \frac{q^Xq^Y}{q^X + (1- \tau)\Delta x }$.

Similarly, when the LT sells $\Delta y$ units of asset $Y$, she receives $\Delta x'$ units of asset $X$ such that $f(q^X - \Delta x', q^Y +(1 - \tau)\Delta y) = f(q^X, q^Y)$, which, for a CPM implies that
$\Delta x' = q^X - \frac{q^Xq^Y}{q^Y + (1- \tau)\Delta y}
$.
% \begin{equation*}
% \Delta x' = q^X - \frac{q^Xq^Y}{q^Y + (1- \tau)\Delta y}
% \end{equation*}

\noindent {\textbf{Concentrated Liquidity.} In CPMs with CL, LPs can concentrate their liquidity over a specific range of rates $(Z^{\ell}, Z^u]$. The rates $Z^{\ell}$ and $Z^u$ take values from a discrete grid
\begin{equation}
\label{eq:grid}
    \{Z^{-N}=\underline Z,Z^{-N+1},\dots,Z^{0},\dots,Z^{N-1},Z^N=\overline Z\}
\end{equation}
where $N \in \mathbb{N}_0$. The smallest possible liquidity range is defined by two consecutive ticks $(Z^i, Z^{i+1}]$ and it is called a \textit{tick range}. Within each tick range, the CPM with CL behaves locally as a CPM.

When an LP opens a new position, she specifies a range of rates $(Z^{\ell}, Z^u]$ and a liquidity depth $\tilde{\kappa}$ (or liquidity units), which determines the quantities $x$ and $y$ in assets $X$ and $Y$ respectively that the LP should provide, namely 
\begin{comment}
$(x,y) = \left(0,\,
\tilde{\kappa}\left((Z^{\ell})^{-1/2}-(Z^{u})^{-1/2}\right)\right)$,
if $Z\leq Z^{\ell}$; $(x,y)=\left(
\tilde{\kappa}\left(Z^{1/2}-(Z^{\ell})^{1/2}\right),\,
\tilde{\kappa}\left(Z^{-1/2}-(Z^{u})^{-1/2}\right)
\right)$,
if $Z^{\ell}<Z\leq Z^{u}$; and, $(x,y)=\left(
\tilde{\kappa}\left((Z^{u})^{1/2}-(Z^{\ell})^{1/2}\right),\,
0
\right)$,
if $Z>Z^{u}$.
\end{comment}

\begin{equation}
\label{eq:uni-v3-positions}
\small
(x,y)=
\begin{cases}
\left(0,\,
\tilde{\kappa}\left((Z^{\ell})^{-1/2}-(Z^{u})^{-1/2}\right)\right),
& Z\leq Z^{\ell}, \\[0.5em]

\left(
\tilde{\kappa}\left(Z^{1/2}-(Z^{\ell})^{1/2}\right),\,
\tilde{\kappa}\left(Z^{-1/2}-(Z^{u})^{-1/2}\right)
\right),
& Z^{\ell}<Z\leq Z^{u}, \\[0.5em]

\left(
\tilde{\kappa}\left((Z^{u})^{1/2}-(Z^{\ell})^{1/2}\right),\,
0
\right),
& Z>Z^{u}.
\end{cases}
\end{equation}

% In words, %based on \eqref{eq:uni-v3-positions}, 
% when the exchange rate is $Z \le Z^\ell$, the position $(Z^\ell, Z^u]$ translates to providing only the asset $Y$, and conversely, when $Z > Z^u$, the position $(Z^\ell, Z^u]$ translates to providing only the asset $X$. When $Z^{\ell} < Z \le Z^{u}$, the position translates to providing both assets $X,Y$, and the quantities are determined by the distance of the current marginal exchange rate from the boundaries of the position $Z^{\ell}$ and $Z^u$ respectively.
To obtain the total liquidity depth in a tick range, one needs to sum the depths of the individual liquidity positions in the same tick range. The formulas in \eqref{eq:uni-v3-positions} define a reparametrization whereby the state of the CPM with CL can be fully described by the total liquidity depth in each tick range and the marginal exchange rate. This representation is more convenient since an LP action changes the liquidity depths of the respective ticks, while an LT action changes the marginal exchange rate $Z$.\footnote{For a more thorough treatment on the mechanics of CPMs with CL, see \cite{Adams21v3, Cartea2024defi}}
% Moreover, when an LT's trade is large enough to move the marginal exchange rate past the current tick range, the trade is split and executed in parts satisfying locally in each tick range the mechanics of CPM described above for the respective depth of each tick range. 
Lastly, the trading fees are stored in each tick range, and are distributed pro-rata to each LP. For example, if the LP's position with depth $\tka$ is in a tick range with total depth $\kappa$, then for every liquidity taking trade that pays in fees an amount of $p$, the LP  earns the amount of 
%\begin{equation}
%
%\label{eq:fee-pro-rata}
$\tilde{p} =  \frac{\tilde{\ka}}{\kappa} p \,  \mathds{1}_{Z^{\ell} < Z \le Z^u}$. 
%\end{equation}                    

\section{An impulse control formulation} %for liquidity provision}
\label{sec:impulse-control-lp}
%%%%%%%%%%%%%%%%%%%%%%%%%%%%%%%%%%%%%%%%%%%%%%%%%%%%%%
Having introduced the mechanics of CPMs with CL, we now turn to the dynamic liquidity provision problem faced by an individual LP. In this problem, the LP observes the market conditions, which are defined in the next paragraph, and decides when to close her position and open a new one, i.e. when to rebalance, and what the new position should be. When the LP rebalances, she collects the accumulated fees from her last position and then reinvests her initial wealth to the new position. 
Moreover, each rebalancing incurs a fixed transaction cost, which captures the gas cost paid for interacting with the blockchain, and for that reason, we formulate the problem as an impulse control problem.\footnote{The difference between an impulse control problem and a continuous stochastic control problem is that, in the former case, the controller takes actions occasionally (interventions) because of the transaction costs, while in the latter case, the controller takes actions continuously; cf. \cite{Eastham88Impulse}.}
More formally, the controller is an LP who, at a sequence of intervention times, chooses a range of ticks from \eqref{eq:grid} to post her liquidity, with $L_t=(L^i_t)_{i=-N}^{N-1}$ denoting the existing liquidity as a vector in the pool, excluding the agent.\footnote{For simplicity we take this liquidity vector to be deterministic; however, our RL implementation supports Markovian stochastic counterparts.} 

We consider that the price formation of the risky asset $Y$ is exogenous to the CPM with CL and we denote its external midprice at time $t$ as $S_t$. The external midprice $S_t$ follows the dynamics
\begin{equation} \label{eq:GBM}
dS_t = \mu S_t dt + \sigma S_t dW_t,\qquad S_0>0,
\end{equation} 
with drift $\mu$, volatility $\sigma>0$, and where $W_t$ is a Brownian motion. Concurrently, at each time $t$, the CPM with CL quotes its own price for the risky asset $Y$, the marginal exchange rate $Z_t$. 
Thus, at each time $t$, there are two prices quoted for the risky asset $Y$; the marginal exchange rate (the price quoted by the AMM) denoted by $Z_t$, and the external price, denoted by $S_t$.

In terms of the dynamics of the marginal exchange rate $Z$, we assume it lives within the grid \eqref{eq:grid}. In particular, when it is at tick $i$, an LT trade can only move it from $Z^i$ to $Z^{i+1}$ or $Z^{i-1}$ provided it stays in the grid. Hence, the marginal exchange rate $Z$ evolves solely through LT trades, which constitute the order flow. More precisely, let $\counts^{+,i}$ and $\counts^{-,i}$ denote the counting processes of trades moving the marginal exchange rate from $Z^i$ to $Z^{i+1}$ and from $Z^{i+1}$ to $Z^i$, respectively. The total number of up/down  movements $\counts^\pm$ is given by $\counts^\pm_t = \sum_i \counts^{\pm,i}_t$.

We work under a probability measure $\mathbb{P}$ under which
the counting processes $\counts^\pm_t$ (counting the number of LT trades that move the marginal exchange rate up/down in the grid \eqref{eq:grid}) have intensities
\begin{equation}
\label{eq:controlled-linear-intensities}
    \lambda^{\pm}_t
    = \max\Big\{\lambda_0,\,
    \lambda_1 \pm \lambda_2(S_t-Z_{t-}) 
    \Big\},
\end{equation}
where $\lambda_0>0$ and $\lambda_1,\lambda_2\ge0$. 
The term $\lambda_1$ captures the baseline order flow (i.e. the noise traders), while the term $\pm \lambda_2(S_t-Z_{t-})$ captures arbitrage pressure. The term $\lambda_0>0$ is a safeguard parameter to prevent intensities from becoming negative.
This model for intensities has been studied previously in \cite{Aqsha2025equilibrium}.
%,bergault2025optimal

The LP's filtration (the LP's observable information when taking new actions) will be generated by $S_t$ and $\counts_t^\pm$, since $\counts_t^{\pm,i}$ is determined by $\counts_t^\pm$ and the marginal exchange rate $Z_t$. 
Hence, the LP's control is an impulse strategy
\begin{equation}
\label{eq:lp-impulse-control}
\textstyle
    \alpha := \big( \imp_n, Z^{\ell}_n, Z^u_n \big)_{n\ge 1},
\end{equation}
where $(\imp_n)_{n\ge1}$ is an increasing sequence of stopping times and, at each intervention time $\imp_n$, the LP chooses a lower and an upper tick satisfying $\underline Z \le Z^{\ell}_n < Z^u_n \le \overline Z$.   

The pair $(Z^\ell_n,Z^u_n]$ specifies the liquidity range until the next intervention. 
Equivalently, if $i_t$ denotes the active tick at time $t$, the LP's liquidity vector is
\begin{equation}
\tilde L^{\alpha,i}_t
    = \tilde\kappa^\alpha_t\,
    \mathbf{1}_{\{Z^\ell_t < Z^{i_t} \le Z^{i_t+1}\le Z^u_t\}},
    \qquad i=-N,\dots,N-1,  
\end{equation}
where $\tilde\kappa^\alpha_t$ is the liquidity depth deployed in the selected range, determined by the LP's available capital at rebalancing and the value per liquidity unit in that range. We then define the total liquidity depth in tick range $(Z^i,Z^{i+1}]$ by $\kappa^{\alpha,i}_t:=L^i_t+\tilde L^{\alpha,i}_t$.

% We write $\mathcal{A}$ for the set of admissible strategies $\alpha$ such that $(\imp_n)_{n\ge1}$ are ordered stopping times, the chosen ranges are predictable at the intervention times, and the cumulative intervention costs are integrable. 

The fees collected by the LP are proportional to her share of liquidity in the current tick. 
Thus, for a fee rate $\fee\in(0,1)$, the cumulative fee process is
\begin{align}
\label{eq:controlled-fee-process}
P^\alpha_t
    &= \textstyle\sum_{i=-N}^{N-1}\left[
    \int_0^t
    \frac{\tilde L^{\alpha,i}_{s-}}{\kappa^{\alpha,i}_{s-}}
    \left(\frac{\fee}{1-\fee}\right)
    \Delta x^i_s \,d \counts^{+,i}_s \nonumber +
    \int_0^t
    \frac{\tilde L^{\alpha,i}_{s-}}{\kappa^{\alpha,i}_{s-}}
    \left(\frac{\fee}{1-\fee}\right)
    S_s\,\Delta y^i_s \,d \counts^{-,i}_s
    \right]\notag
\end{align}
where $\Delta x^i_s$ and $\Delta y^i_s$ are the trade sizes required to move the marginal exchange rate by one tick in the corresponding direction. 
The dependence on the chosen width enters through $\tilde L^{\alpha,i}$: 
for fixed amounts of capital, a narrower range creates a larger liquidity share in the active tick, while a wider range lowers the share but remains active in more price states. In our setup, each intervention incurs a fixed transaction cost; 
this is a key distinction from the frictionless case in ~\cite{Cartea2024defi}. When the LP changes her position at time $\imp_n$, she pays a fixed gas cost. 
% and a proportional cost generated by the change in the token composition of her position
We write the cumulative cost up to time $t$ as
% \begin{equation}
% \label{eq:impulse-lp-cost}
%     C^\alpha_t=
%     \textstyle\sum_{\imp_n\le t}
%     \left(g+\zeta Z_{\imp_n}|\Delta Y^\alpha_{\imp_n}|
%     \right), \notag
% \end{equation}
$C^\alpha_t= \textstyle\sum_{\imp_n\le t}g$, where $g> 0$ is the gas cost for each new LP action.
%$C^\alpha_t=\textstyle\sum_{\imp_n\le t}g, \notag$ where $g>0$ is the gas cost for each new LP action.
% , $\zeta\ge0$ is a proportional execution cost parameter, and $\Delta Y^\alpha_{\imp_n}$ is the change in the amount of risky asset $Y$ induced by the rebalance

Let $V^\alpha_t$ be the value of the LP's current position in units of asset $X$ at the current external price, i.e. $V^\alpha_t = x^\alpha_t+S_ty^\alpha_t\,.$
Relative to holding the initial token amounts outside the pool, the realised IL term is $\mathrm{IL}^\alpha_t=-\Big[V^\alpha_t-\big(x^\alpha_0+S_ty^\alpha_0\big)\Big]$. The risk-neutral LP solves
\begin{equation}
\label{eq:lp-impulse-risk-neutral}
\sup_{\alpha\in\mathcal{A}}
\mathbb{E}^{\mathbb{P}^\alpha}
\left[
P^\alpha_T
- \mathrm{IL}^\alpha_T
- C^\alpha_T
\right].
\end{equation}
where $P^\alpha_T$ denotes collected fees, $\mathrm{IL}^\alpha_T$ denotes impermanent loss, and $C^\alpha_T$ denotes rebalancing costs.

% Equivalently, we also consider the case of risk aversion through exponential utility for which we solve
% \begin{equation}
% \label{eq:lp-impulse-risk-averse}
%     \sup_{\alpha\in\mathcal{A}}
%     \mathbb{E}^{\mathbb{P}^\alpha}
%     \left[
%         -\exp\left(
%             -\gamma
%             \big(
%                 P^\alpha_T
%                 -
%                 \mathrm{IL}^\alpha_T
%                 -
%                 C^\alpha_T
%             \big)
%         \right)
%     \right],
%     \qquad \gamma>0 .
% \end{equation}

The problem therefore captures the main trade-off faced by a LP in a CPM with CL; narrow ranges increase the fee share conditional on being active, but expose the position to concentration risk and more frequent rebalancing, while wider ranges lower the fee share but reduce the probability of falling out of range. 
The impulse structure makes the rebalancing friction explicit. 
% , and the controlled intensities in \eqref{eq:controlled-linear-intensities} allow the LP's liquidity placement to affect future order flow.

\noindent \textbf{Risk Aversion through Running Inventory Penalty.} Besides the risk neutral case defined in \eqref{eq:lp-impulse-risk-neutral}, we will modulate the performance of the LP under different levels of risk aversion. 
The simplest way to introduce risk aversion is through an inventory penalty. We therefore extend the objective function in \eqref{eq:lp-impulse-risk-neutral} as follows
\begin{equation} 
\label{eq:running_inv}
           \sup_{\alpha\in\mathcal{A}}
    \mathbb{E}^{\mathbb{P}^\alpha}\left[{ P^\alpha_T }
        -
        {\mathrm{IL}^\alpha_T}
        -
        {C^\alpha_T} - \pen \int_0^T(y_t^\alpha - \hat{y})^2dt\right],
\end{equation}
where the constant $\pen \in \mathbb{R}^+$ determine the running penalty deviating from some level $\hat{y}$ (see \cite{Cartea23AMMDesigns}). In particular, the LP's performance criterion \eqref{eq:running_inv} aims to maintain a target reserve level $\hat{y}$. Note that without a target reserve $\hat{y}$, the performance criterion would essentially force the LP to liquidate the risky asset that she holds.
It is well-known in the algorithmic trading literature (as shown in ~\cite{Cartea17Uncetrainty}) that the running penalty term in \eqref{eq:running_inv} can be interpreted as arising from the agent's ambiguity aversion with respect to the external price $S_t$.
\begin{comment}
\footnote{There are other well studied forms of adding calibrated risk aversion, e.g., the exponential utility $    \sup_{\alpha\in\mathcal{A}} \mathbb{E}^{\mathbb{P}^\alpha}
    \left[
        -\exp\left(
            -\gamma
            \big(
                P^\alpha_T
                -
                \mathrm{IL}^\alpha_T
                -
                C^\alpha_T
            \big)
        \right)
    \right], $
    where $\gamma>0$ is the risk aversion level.} 
\end{comment}
As we show in \cref{sec:experiments}, the per-step running inventory penalty is the main component pushing the LP towards wider positions.

% \paragraph*{Risk Aversion through exponential utility} 
% Another way in which we consider risk aversion is through a parameter $\gamma$ in the optimisation of exponential utility  
% \begin{equation}
% \label{eq:lp-impulse-risk-averse}
%     \sup_{\alpha\in\mathcal{A}}
%     \mathbb{E}^{\mathbb{P}^\alpha}
%     \left[
%         -\exp\left(
%             -\gamma
%             \big(
%                 P^\alpha_T
%                 -
%                 \mathrm{IL}^\alpha_T
%                 -
%                 C^\alpha_T
%             \big)
%         \right)
%     \right],
%     \qquad \gamma>0 .
% \end{equation}
\noindent \textbf{Theoretical Results.}\label{sec:theory}
Two structural properties of the LP's PnL follow directly from the CL mechanics of
Section~\ref{sec:prelims}, and both are used to interpret the learned policies in
Section~\ref{sec:experiments}. 
 
\begin{lemma}[Fee revenue is linear in own liquidity]
\label{lem:linear-fee-revenue-liquidity}
Under one-tick jumps, the fees collected by an individual LP are linear in $\tka$ and
independent of the other LPs liquidity $L^i$ in the pool.
\end{lemma}
 
\begin{proof}[Proof Sketch]
Moving from $Z^i$ to $Z^{i+1}$ requires an after-fee trade of size
$\Delta x^i=\ka^{i}\big((Z^{i+1})^{1/2}-(Z^{i})^{1/2}\big)$ in $X$, and symmetrically
$\Delta y^i=\ka^{i}\big((Z^{i})^{-1/2}-(Z^{i+1})^{-1/2}\big)$ in $Y$ in the opposite direction: the volume needed to cross a tick is proportional to the \emph{total} depth $\ka^{i}$ of that tick. Since the LP earns the pro-rata share $\tka^{i}/\ka^{i}$ of the resulting fee, $\ka^{i}$ cancels in the LP's fee revenue process $P^\alpha_t$, leaving an integrand proportional to $\tka^{i}_{s-}$ and free
of $L^i$. Deeper liquidity in a range absorbs more volume per tick crossed, and the pro-rata rule returns to the LP exactly her share of that extra volume.
\end{proof}

%\noindent \textcolor{black}{This invariance to competing liquidity is specific to the one-tick-crossing order-flow representation adopted here. Under an alternative model in which incoming order quantity is primitive, competing liquidity would generally affect price impact, tick-crossing probabilities, and hence fee revenue.}

\begin{lemma}[IL is linear in own liquidity]
\label{lem:linear-position-change-liquidity}
For a position of depth $\tka$ in $(Z^{\ell},Z^u]$, the IL between posting and withdrawal is linear in $\tka$.
\end{lemma}
 
\begin{proof}[Proof Sketch]
Every branch of \eqref{eq:uni-v3-positions} is positively homogeneous of degree one in $\tka$, so
$x_t$, $y_t$, and hence $\IL_t=-\big[x_t+y_tS_t-(x_0+y_0S_t)\big]$, are linear in $\tka$, whether or not the position is in range. While $Z_t\in(Z^{\ell},Z^u]$ the boundary terms cancel and $\IL_t=\tka\big[Z_0^{1/2}-Z_t^{1/2}+S_t\big(Z_0^{-1/2}-Z_t^{-1/2}\big)\big]$, so the range enters only through $\tka$. For fixed capital, a narrower range implies a larger $\tka$, which makes the IL curve locally steeper, showing the cost side of concentration.
\end{proof}

\noindent From \Cref{lem:linear-fee-revenue-liquidity} and \Cref{lem:linear-position-change-liquidity}, we thus have:

\begin{theorem}[Affinity of the LP's PnL in liquidity]\label{theorem:linearity}
The LP's PnL $P^{\al}_T-\IL^{\al}_T-C^{\al}_T$ is affine in $\tka$: it is linear through the fee and IL terms, while the gas cost $C^{\al}_T$ enters as a $\tka$-independent constant.
\end{theorem}
These results assume an LP whose liquidity is small relative to the pool, so that her actions do not affect the liquidity taking activity in the future.%total order flow.

\noindent \textbf{Learning Agents as Dynamic LPs.} The impulse control problems in \eqref{eq:lp-impulse-risk-neutral} and \eqref{eq:running_inv} are analytically intractable in realistic settings, since the LP must jointly decide when to intervene and where to deploy liquidity under stochastic prices, order arrivals, gas costs, inventory exposure, and a discrete tick grid. RL provides a natural numerical approach for learning admissible strategies of the form \eqref{eq:lp-impulse-control}, with rewards defined as sample-path counterparts of the objectives inside the supremum. For a risk-neutral LP, the reward is realised PnL from fees, IL, and rebalancing costs, corresponding to \eqref{eq:lp-impulse-risk-neutral}; for a risk-averse LP, we use the inventory-penalised criterion in \eqref{eq:running_inv}, where $\pen$ controls the degree of risk aversion.

The mechanics and frictions of CL, together with \Cref{theorem:linearity}, imply the key economic trade-off studied in our experiments. Narrower ranges increase fee income conditional on being active but expose the LP to concentration risk and costly rebalancing; wider ranges remain active over more price states and reduce out-of-range and inventory risk, but dilute liquidity across ticks and lower the LP’s fee share per tick. In the risk-neutral case, the linearity of fee revenue and IL in the LP’s own liquidity implies a strong incentive to concentrate liquidity in the expected most profitable tick range, providing a theoretical benchmark for studying risk aversion, inventory penalties, and rebalancing frictions.

We therefore consider two action parametrisations. A flexible agent that freely chooses the lower and upper bounds $(Z^\ell_n,Z^u_n]$ and a narrow agent with the minimum fixed position width of two ticks $(Z^i_n,Z^{i+2}_n]$ that learns only where to place her liquidity. The narrow specification isolates the value of learning \emph{where} to post liquidity from the value of learning \emph{how wide} the position should be. Decisions are made only at intervention times, consistent with the impulse-control formulation and with the fact that continuous rebalancing is unrealistic under blockchain frictions. 

\section{Experimental Results} 
\label{sec:experiments}

Based on the above, in the following experiments, we implement two RL agents with the Proximal Policy Optimisation (PPO) algorithm, namely the PPO and PPO\_narrow. For the action space, when $(Z^\ell_{n+1},Z^u_{n+1}]=(Z^\ell_{n},Z^u_{n}]$, the agent chooses to hold her previous position, and thus no extra gas costs are charged. For the first action, if $(Z_0^\ell, Z_0^u]=(0,0]$ the agent chooses not to provide liquidity. The observation space contains the main variables needed for liquidity provision: mispricing $S_t-Z_t$, distances from the current pool price to the position boundaries, time, gas cost, cumulative fee income, and the LP’s current holdings in the risky and numéraire assets. 

In our experiments, each trained policy and benchmark strategy is evaluated on 1000 independent simulations using random seeds distinct from those used during training. We compare strategies across scenarios and profile the learned policies as state-dependent functions, relating their behaviour to empirical liquidity-provision patterns. The aim is to understand how learned LP policies adapt their position width, asymmetry, and rebalancing frequency to the LP’s risk appetite and to the stochastic and blockchain-specific properties of the environment. 
% All results are reproducible, and the anonymised code repository is available at https://shorturl.at/B8B81. 

\textbf{Experimental Setup:} We train and evaluate PPO and PPO\_narrow, in the stochastic environment of \Cref{sec:impulse-control-lp}. The external price of the risky asset $Y$ (e.g. ETH) follows the geometric Brownian motion in \eqref{eq:GBM}. The liquidity-taking arrivals follow the linear intensity model in \eqref{eq:controlled-linear-intensities}, for which we set $\lambda_1=15$ and $\lambda_2=4000$. These values represent an order-flow intensity that is mainly driven by arbitrage pressure from deviations between the pool price and the external price. Equations (\ref{eq:GBM}) and (\ref{eq:controlled-linear-intensities}) are the two sources of stochasticity in the experiments. For the external price, we set $S_0=1000$, $dt=0.001$ (assuming $T=1$), $\mu=0$, and vary volatility over $\sigma \in {\{0.005,0.01,0.015,0.02,0.025,0.03\}}$. Higher $\sigma$ increases the likelihood of sharp price movements, IL, and positions moving out of range. The main blockchain friction is the gas cost paid at each rebalance, for which we consider $g \in {\{2,4,6\}}$, corresponding to low, medium, and high operational costs. We train both risk-neutral and risk-averse LPs: risk-neutral agents use the cumulative PnL reward in \eqref{eq:lp-impulse-risk-neutral}, while risk-averse agents use the inventory-penalised reward in \eqref{eq:running_inv}, with $\pen \in \{20,50,80\}$ controlling the strength of inventory risk aversion. As observed in practice, the pool price $Z_t$ follows closely $S_t$ due to arbitrage. Therefore, we set the parameters of $S_t$ within the range of empirical estimations calculated in \cite{Cartea2024defi}. The rest of the parameters were calibrated to be economically sensible based on the above within our model-based setup. The action range is restricted to $[-15,15]$ ticks around the current tick. Each LP starts with 1000 numéraire units, and the proportional fee tier is set to $\fee=0.3\%$, a standard UniswapV3 fee tier for token pairs.

We benchmark PPO and PPO\_narrow against rule-based and model-based strategies operating within the same $[-15,15]$ tick window. The \textit{DeployOnceAgent} deploys capital at the start and never rebalances. We use two variants: \textit{DeployNarrow}, which deploys over $[-1,1]$ around the initial pool price, and \textit{DeployWide}, which deploys over $[-15,15]$. The former represents the narrow single-shot allocation motivated by the concentration incentives in \Cref{sec:theory}, while the latter maximises the probability of remaining active throughout the simulation. The \textit{ArrivalRebalanceAgent} also posts liquidity over $[-1,1]$, but rebalances after every $N_{arr}$ liquidity-taking arrivals, adding a simple timing rule relative to \textit{DeployNarrow}. Finally, we compare against the \textit{CarteaDrissiMongaAgent} (CDM), a model-based closed-form strategy for concentrated AMMs based on \cite{Cartea2024defi}. 

Each trained RL agent corresponds to an independent environment instance, giving 72 scenarios in total ($6\cdot3\cdot4$) across $\sigma$, $g$, and risk-aversion settings. To ensure statistical robustness across the training variability of the RL algorithms, each scenario is assessed across 10 different training seeds. To reduce policy-gradient variance and expose policies to a broad state distribution, we collect rollouts using $N_{\mathrm{traj,train}}=200$ parallel trajectories advanced synchronously. Each trajectory has $N_{\mathrm{steps}}=1000$ environment steps, with intervention times every $K_{\mathrm{dec}}=100$ steps. Thus, each rollout contains $N_{\mathrm{traj,train}}N_{\mathrm{steps}}/K_{\mathrm{dec}}=2000$ agent-decision transitions. We set $E_{\mathrm{train}}=300$, yielding a total training budget of $6\times10^5$

\begin{table}[t]
\centering
\caption{Mean PnL and $5\%$ CVaR in USDC over 1000 common evaluation trajectories. Values following $\pm$ denote the standard deviation across 10 independently trained policies. Bold marks the best RL and benchmark method. $p$ tests their difference, treating training seeds as replication units and additionally propagating trajectory sampling error.}\vspace{-0.15cm}
\label{tab:main1}
\resizebox{0.475\textwidth}{!}{%
\begin{tabular}{@{}l*{4}{c}@{}}
\toprule
& \multicolumn{2}{c}{$\sigma=0.01,\ g=2$}
& \multicolumn{2}{c}{$\sigma=0.02,\ g=4$}\\
\cmidrule(lr){2-3}\cmidrule(lr){4-5}
Agent & PnL & CVaR & PnL & CVaR\\
\midrule
DeployNarrow     & 15.85          & -14.52 & 6.48          & -35.37\\
DeployWide       & 12.78          & -15.11 & 4.09                    & -36.48\\
ArrivalRebalance & \textbf{49.72} & 6.86   & 5.69                    & -36.23\\
CDM              & 35.18          & \textbf{11.55}  & \textbf{15.01}                   & \textbf{-26.23}\\
\midrule
&\multicolumn{4}{c}{Risk neutral}\\
\cmidrule(l){2-5}
PPO        & 42.31$\pm$0.97          & 6.08$\pm$0.97   & 5.33$\pm$0.65           & -36.51$\pm$0.77\\
PPO\_narrow & \textbf{49.91$\pm$0.38} & 9.15$\pm$0.65   & \textbf{6.51$\pm$0.20}  & -35.69$\pm$0.42\\
\cmidrule(l){2-5}
$p$        & 0.81                    & 0.0014          & $<10^{-10}$                  & $<10^{-4}$\\
\midrule
&\multicolumn{4}{c}{Risk averse ($\phi=50$)}\\
\cmidrule(l){2-5}
PPO        & 30.36$\pm$1.60 & 3.07$\pm$1.56          & 1.49$\pm$0.16 & -9.14$\pm$0.46\\
PPO\_narrow & 46.11$\pm$0.67 & \textbf{9.32$\pm$0.74} & 2.27$\pm$0.40 & \textbf{-8.86$\pm$0.87}\\
\cmidrule(l){2-5}
$p$        & $<10^{-10}$    & 0.04                   & $<10^{-10}$   & $<10^{-4}$\\
\bottomrule
\end{tabular}%
}
\end{table}

\begin{table}[t]
\centering
\caption{Under adversarial conditions ($\sigma=0.03$, $g=6$), the risk-averse RL agents learn mainly not to deploy.}\vspace{-0.15cm}
\label{tab:main2}
\resizebox{0.475\textwidth}{!}{%
\begin{tabular}{@{}l*{4}{c}@{}}
\toprule
& \multicolumn{2}{c}{$\sigma=0.03,\ g=6$}\\
\cmidrule(l){2-3}
Agent & PnL & CVaR\\
\midrule
DeployNarrow     & \textbf{0.10} & \textbf{-56.1}7\\
DeployWide       & -2.12         & -57.33\\
ArrivalRebalance & -5.38         & -62.64\\
CDM              & -3.10         & -56.93\\
\midrule
&\multicolumn{2}{c}{Risk neutral}&\multicolumn{2}{c}{Risk averse ($\phi=50$)}\\
\cmidrule(lr){2-3}\cmidrule(l){4-5}
PPO        & -1.42$\pm$0.80          & -54.78$\pm$6.61& \textbf{0.00$\pm$0.04}  & -0.68$\pm$1.47\\
PPO\_narrow & -0.06$\pm$0.56 & -52.68$\pm$7.73& -0.01$\pm$0.03 & \textbf{-0.16$\pm$0.51}\\
\cmidrule(l){2-5}
$p$        & 0.64                    & 0.19  & 0.88            & $<10^{-4}$\\
\bottomrule
\end{tabular}%
}
\end{table}

\noindent agent-decision transitions for the PPO optimisation objective, after which average returns stabilise without material performance gains from further training. For the RL architecture we set: learning rate $=3\times10^{-4}$, batch size $= 200$, $\gamma=0.99$, GAE $\lambda=0.95$, clip range $= 0.15$, entropy $= 0.03$ and two hidden MLP layers of 256 units.

\textbf{Performance of Learned Policies: }Tables~\ref{tab:main1} and~\ref{tab:main2} compare the performance of the RL agents and the benchmark strategies across

\begin{figure}[t]
\centering
\vspace{-0.15cm}
\includegraphics[width=0.9\columnwidth]{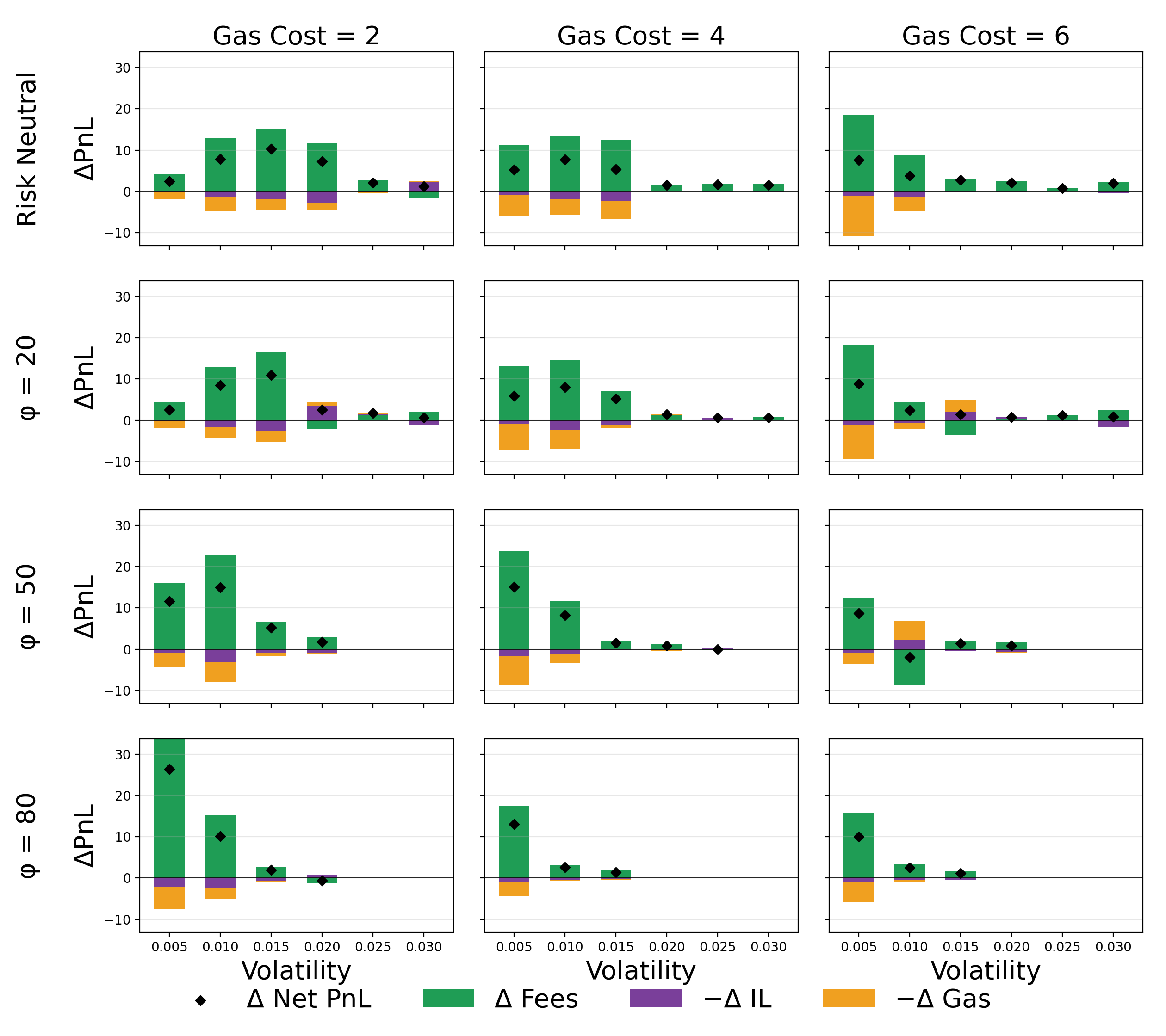}
    \caption{Differences in the mean decomposed PnL components between the RL agents across volatility, by risk profile (rows) and gas cost (columns). $\Delta \mathrm{Net\ PnL}=\Delta \mathrm{Fees}+(-\Delta \mathrm{IL})+(-\Delta \mathrm{Gas})=\text{{\normalfont PPO\_narrow(PnL)}}-\text{{\normalfont PPO(PnL)}}$ denotes the difference in mean PnL and is shown by the diamond marker. A positive $\Delta \mathrm{Fees}$ (green), negative $-\Delta \mathrm{IL}$ (purple), and negative $-\Delta \mathrm{Gas}$ (orange) indicate that PPO\_narrow earns more fees, incurs greater IL, and pays more gas than PPO, respectively. Each quantity is averaged over 1,000 trajectories. Missing bars indicate that one or both agents did not deploy.}
    \label{fig:pnl_sub_b}
%Panel~(\subref{fig:pnl_sub_a}) shows mean PnL, while Panel~(\subref{fig:pnl_sub_b}) decomposes the RL agents' PnL difference into fees, IL, and gas costs.}
\label{fig:collective_pnl_plots}\vspace{-0.15cm}
\end{figure}

\noindent representative market settings chosen as a coherent regime ladder, corresponding to low, medium, and high combinations of $\sigma$ and $g$, under both risk neutrality and an intermediate level of risk aversion ($\phi=50$). The significance tests for PnL use a hierarchical paired $t$-test and for CVaR, a paired bootstrap over both seeds and trajectories ($10^4$ replicates). CDM does not directly account gas costs. Thus, for an optimistic frictionless analytical benchmark version against the RL agents, we consider that the strategy rebalances every $K_{dec}$ free of costs. Its risk aversion lies in its formulation that depends on $\sigma$. For the ArrivalRebalance agent, we finetuned $N_{arr}$ for each scenario from a gridsearch ranging from $N_{arr}=10$ to $800$, optimising both PnL and CVaR (Conditional Value at Risk).

Under the low-volatility, low-gas setting, for the risk neutral case, PPO\_narrow achieves similar PnL with ArrivalRebalance, however with higher CVaR than it, resulting in a better overall strategy. CDM achieves slightly improved CVaR than the risk averse agent, however the PnL of PPO\_{narrow} is better even though gas costs were paid for its positions. As volatility and gas costs increase, profitability declines across all strategies, and the mean-PnL advantage of the RL agents becomes less pronounced. \par
CDM outperforms (with significant $p$ values) in PnL both risk-neutral and risk averse agents under the medium volatility condition. However, the risk aversion case clearly demonstrates the benefits of the RL agents under these market conditions. At $\sigma=0.02$ and $g=4$, the risk-averse RL agents achieve substantially improved CVaR than CDM. Relative to their risk-neutral counterparts under the same market conditions, they sacrifice approximately 4 USDC 

\begin{figure}[t]
\centering
\begin{subfigure}{\columnwidth}
    \centering
    \includegraphics[width=0.86\columnwidth]{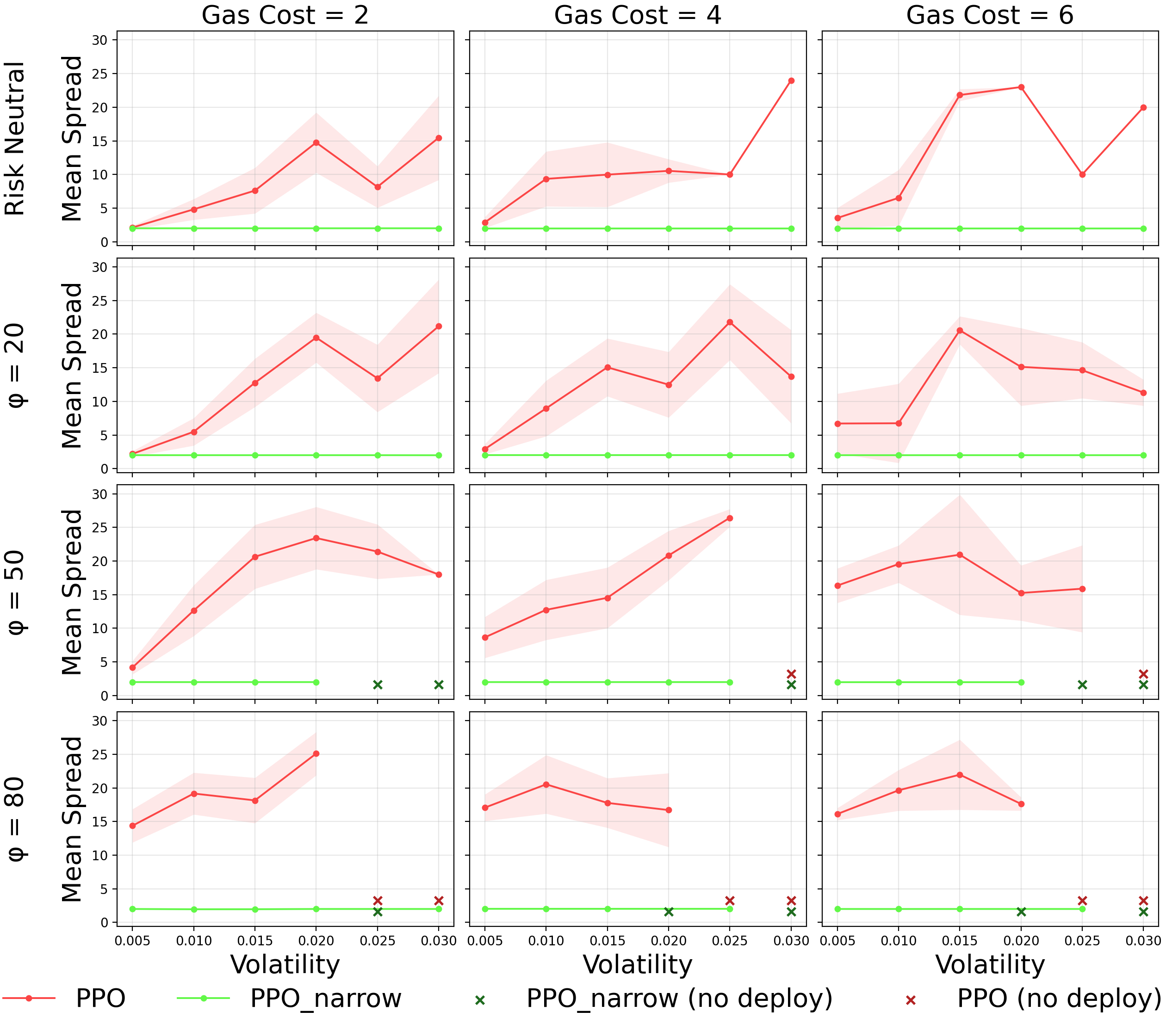}
    \caption{Mean Spreads in ticks of the two RL agents. The shaded area corresponds to the $\pm1$ standard deviation around the mean.} %PPO\_narrow's width is fixed at 2 ticks size.
    \label{fig:width_sub_a}
\end{subfigure} \vspace*{0.07cm}

\begin{subfigure}{\columnwidth}
    \centering
    \includegraphics[width=0.86\columnwidth]{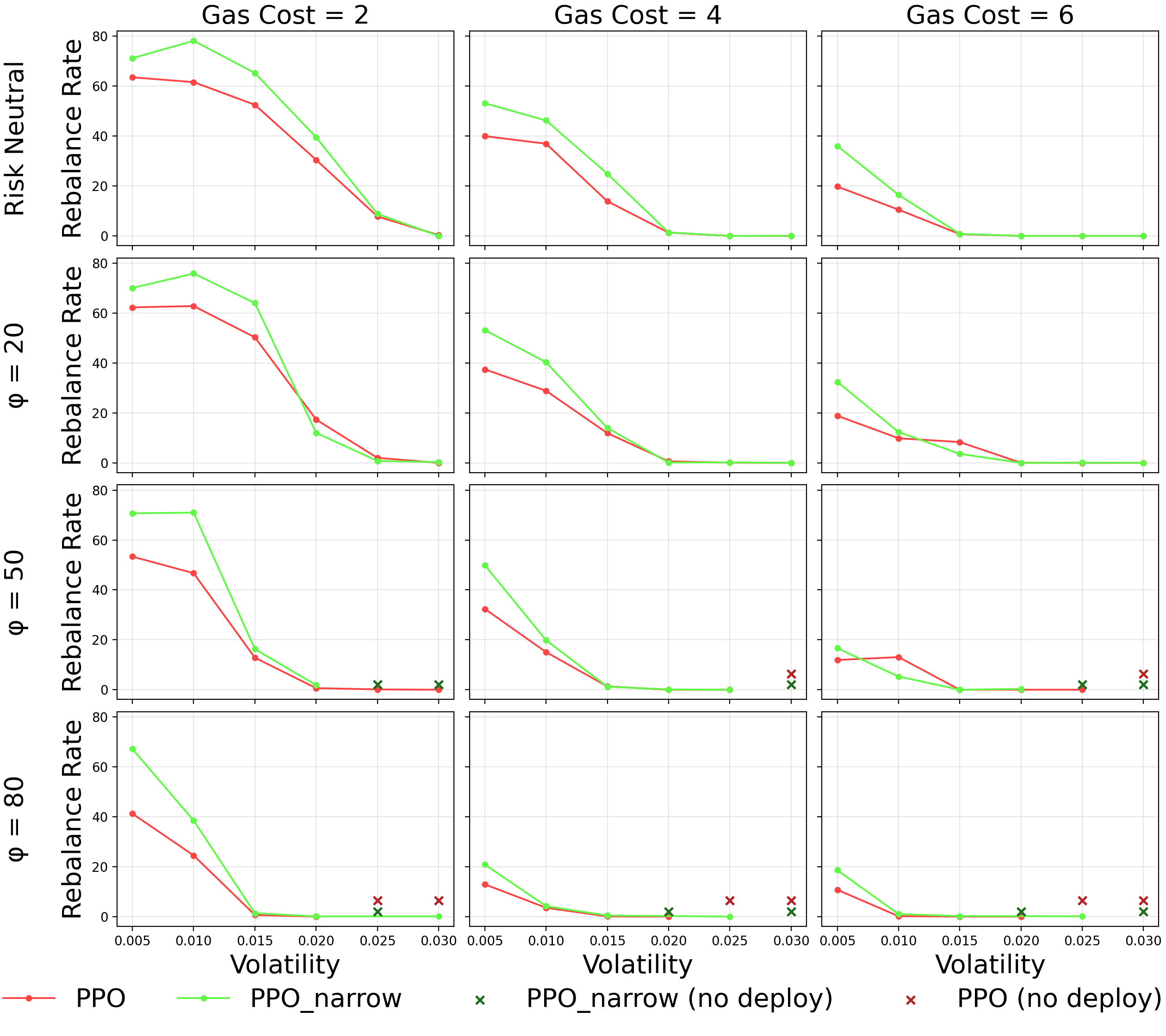}
    \caption{Rebalancing Rate (\%) of the RL agents.}
     \label{fig:rebal_sub_b}
\end{subfigure}
\caption{Position width and rebalancing behaviour of PPO and PPO\_narrow. Panel~(\subref{fig:width_sub_a}) shows mean quoted width, while Panel~(\subref{fig:rebal_sub_b}) shows per-decision rebalance rates. Labels and Legends as in \Cref{fig:collective_pnl_plots}.}
\label{fig:width_rebal}\vspace{-0.35cm}
\end{figure}

\noindent in mean PnL while improving CVaR by approximately 27 USDC, illustrating their ability to trade modest expected profitability for substantially stronger downside protection, reducing extreme losses as market conditions deteriorate. At $\sigma=0.03$ and $g=6$, the risk-averse agents learn mostly not to deploy liquidity, approximating a "DoNothing" benchmark that would report PnL$=0$ and CVaR$=0$. Thus, under sufficiently adverse conditions, the economically meaningful learned behaviour is endogenous market exit rather than superior active liquidity management. Finally, the relatively small standard deviations across independently trained policies indicate stable training outcomes. 

As discussed, PPO\_narrow generally achieves better performance than PPO because its fixed narrow range gives it the theoretically favourable concentration implied by \Cref{theorem:linearity}. However, by decomposing the PnL of the agents, we notice that this advantage comes with higher IL and gas costs (\Cref{fig:pnl_sub_b}), since narrow positions require more frequent rebalancing. Therefore, when risk \noindent aversion is taken into account; by learning wider positions, PPO has lower fee intensity but less active risk management. This mechanism is confirmed in Figure \ref{fig:width_rebal}. PPO systematically increases its average width as volatility, gas costs, and risk aversion rise, provided that

\begin{figure}[t]
\centering
\includegraphics[width=0.8\columnwidth]{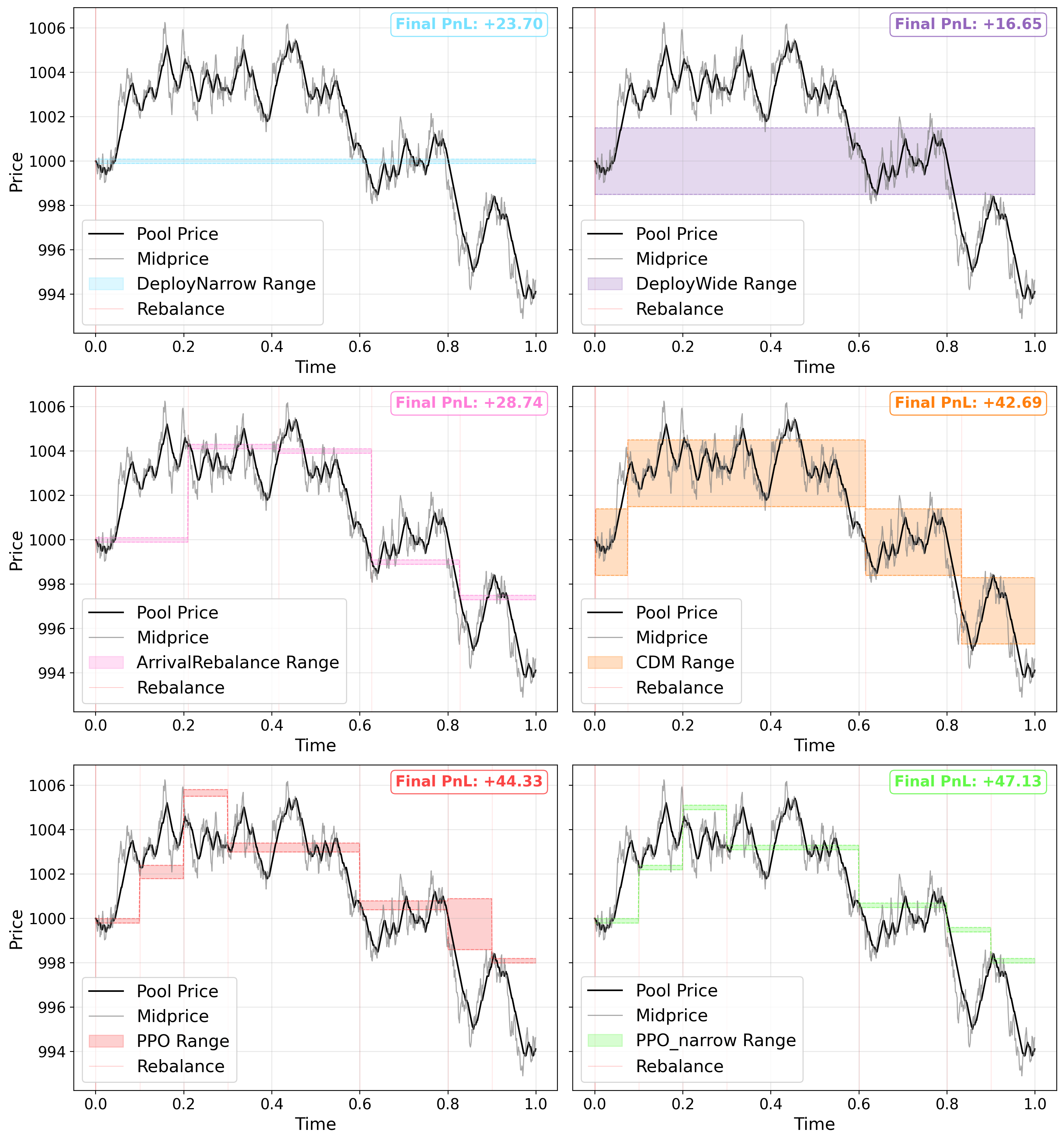}
\caption{Liquidity provision behaviour of all strategies for one random simulation from the evaluation trajectories (scenario: $\sigma=0.01, g=2,$ risk neutral agents). The red thin vertical lines show when agents rebalance.}
\label{fig:price_evolution}
\end{figure}

\begin{figure}[t]
\centering
\includegraphics[width=1\columnwidth]{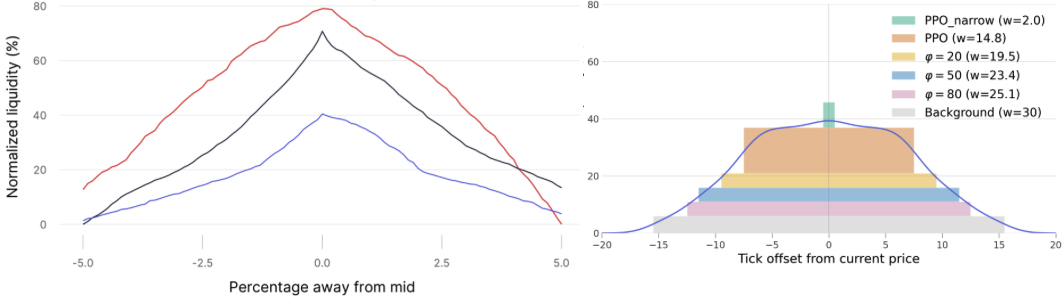}
\caption{Illustration of the empirical CL shapes around pool prices (left) in our model-based environment (right). Left figure shows the normalised liquidity distribution around the current tick for Optimism (red), Arbitrum (black) and Ethereum (purple) as observed from real on-chain data calculated in \cite{adams2024layer2layer2}. Right figure shows a superposition of all four PPO risk-profile agents, PPO\_narrow, and a wide background liquidity. Each agent's position is her reported average width for $\sigma=0.02$ and $g=2$ with indicative relative wealths.}
\label{fig:triangles}
\end{figure}

\noindent the average rebalancing rate is not near zero. Both agents reduce their rebalancing frequency across these dimensions. PPO\_narrow rebalances more often on average, reflecting the greater management required by fixed concentrated positions. Note that PPO\_narrow incurs lower IL and gas costs in the few instances where PPO had a higher rebalance rate, as shown in Figure \ref{fig:rebal_sub_b}. 
\Cref{fig:price_evolution} provides an example of how agents provide liquidity, for high level intuition. \par 
% Additionally, to provide some intuition on how agents react to market dynamics, we illustrate trajectories sampled from the evaluation period in \Cref{fig:price_evolution}. \par

\textbf{A Qualitative Illustration of Aggregate Liquidity Shapes:} The results of \Cref{fig:width_sub_a} show that the learned width of the PPO agent varies systematically with the LP’s risk profile and the frictions of the environment. This provides a natural qualitative interpretation of the liquidity concentration observed empirically around the pool prices in real concentrated AMM pools as the aggregate result of heterogeneous LPs with different risk appetites.

%Normalised Liquidity (\%) shows the existing liquidity in a certain tick range as a percentage of the total liquidity in the pool. 

\begin{comment}
\begin{figure}[t]
\centering
\begin{subfigure}{\columnwidth}
    \centering
    \includegraphics[width=1\columnwidth]{figures/triang_merged.png}
    \caption{Normalised liquidity distribution around current tick for Optimism (red), Arbitrum (black) and Ethereum (purple) as observed from real on-chain data. The distribution was calculated around the current pool price every 15 minutes from 03-01-2023 to 01-01-2024. Image is from \cite{adams2024layer2layer2}.}
    \label{fig:real_triangle}
\end{subfigure}
\begin{subfigure}{\columnwidth}
    \centering
    \includegraphics[width=0.75\columnwidth]{figures/10_triangular_liquidity_2.png}
    \caption{A superposition of all four PPO risk-profile agents, the narrow PPO variant, and a wide passive “background” liquidity. Each agent's position is her reported average width for $\sigma=0.02$ and $gas\_cost=2$ with varying relative wealths. Normalised Liquidity (\%) shows the existing liquidity in a certain tick range as a percentage of the total liquidity in the pool.}
    \label{fig:sim_triangle}
\end{subfigure}
\caption{Conceptual replication of empirical CL shapes around pool prices in our model-based environment.}
\label{fig:triangles}
\end{figure}
\end{comment}

\begin{figure*}[t]
\centering
\begin{subfigure}[t]{0.3\textwidth}
    \centering
    \includegraphics[width=\linewidth]{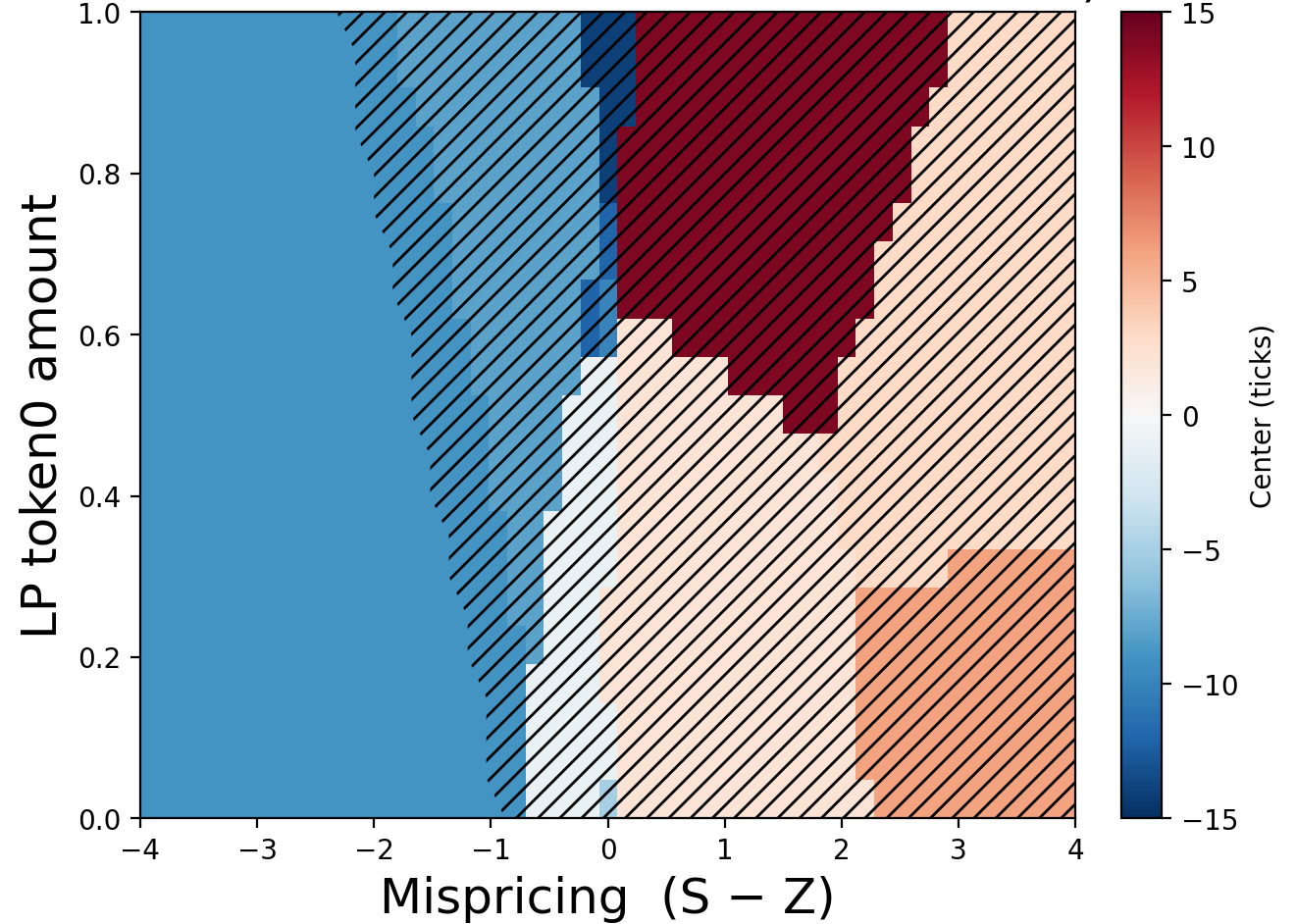}
    \caption{The agent repositions liquidity closer to the external price due to the expected arbitrage flow.}
    \label{subfig:center_ppo_ra}
\end{subfigure}
\hfill
\begin{subfigure}[t]{0.3\textwidth}
    \centering
    \includegraphics[width=\linewidth]{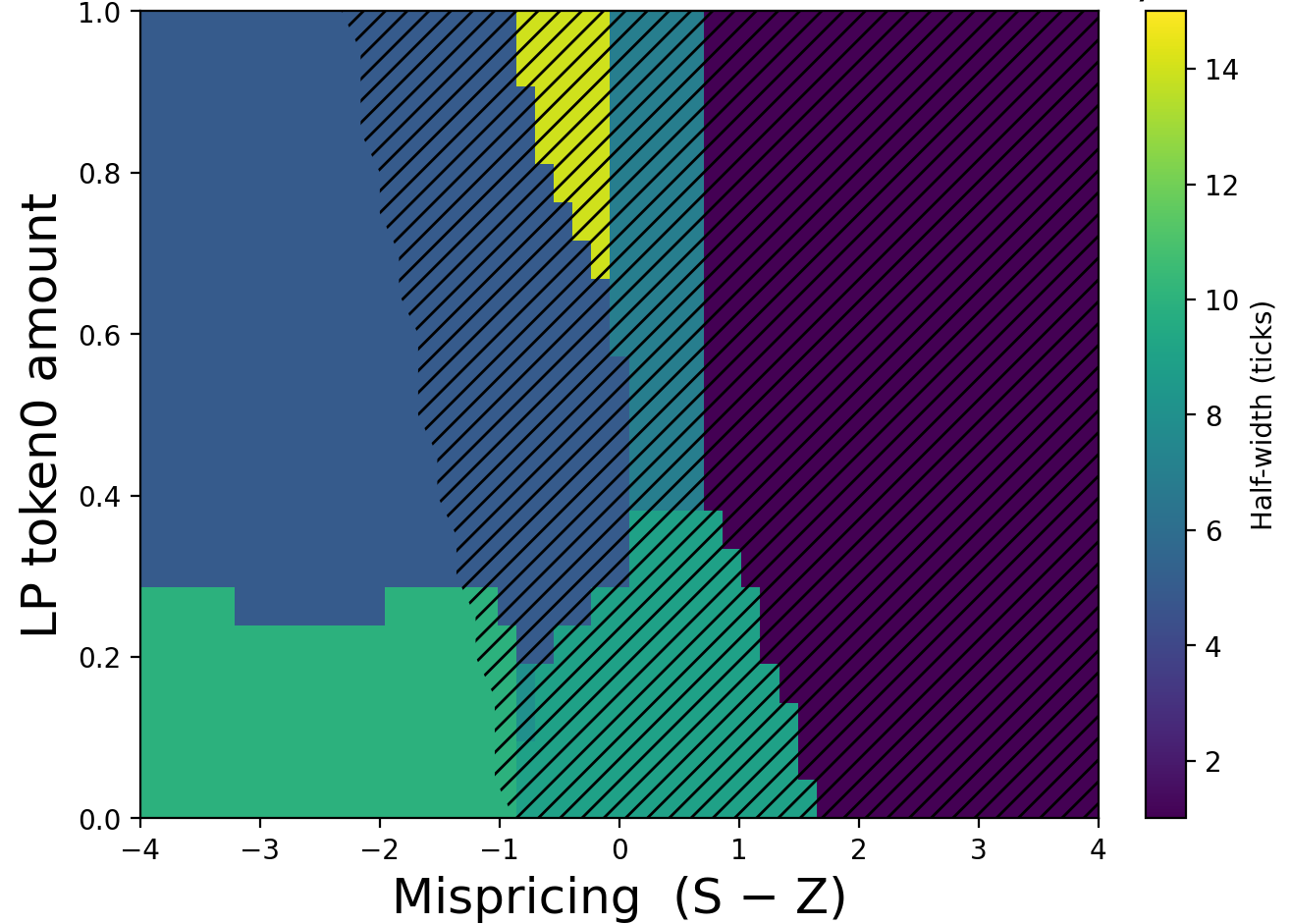}
    \caption{The agent's size of the chosen width depends on the amount of risky token held, functioning as a risk-management device.}
    \label{subfig:width_ppo_ra}
\end{subfigure}
\hfill
\begin{subfigure}[t]{0.3\textwidth}
    \centering
    \includegraphics[width=\linewidth]{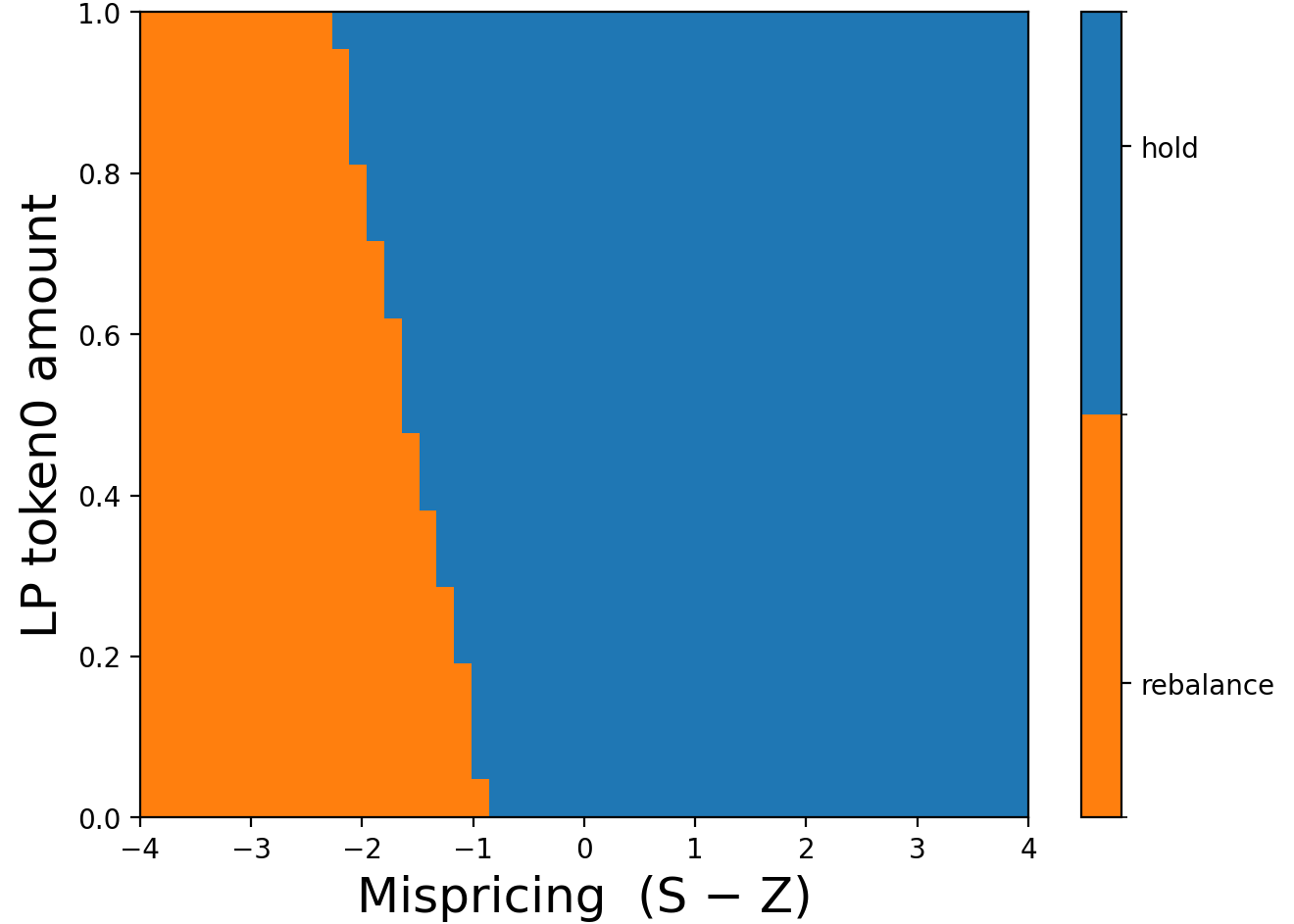}
    \caption{For $S-Z>0$ the agent chooses to hold her position to get rid of the risky asset from the expected incoming trades.}
    \label{subfig:hold_ppo_ra}
\end{subfigure}\vspace{-0.1cm}
\caption{The full action response that a risk-averse PPO agent has for fixed wealth, with $\phi=50$ for $gas\_cost=4$ and $\sigma=0.01$ at $t=0.4$, across two important dimensions: the amount of the risky asset she holds (e.g. ETH), and the adversarial signal (mispricing). It is assumed that the current action of the agent is a symmetrical narrow position (-1,1) around the current tick. The shaded areas in the first two plots are where the agent holds her existing position and the chosen actions can be ignored.}
\label{fig:ppo_ra_actions}\vspace{-0.15cm}
\end{figure*}
\begin{figure}[t]
\centering
\includegraphics[width=\columnwidth]{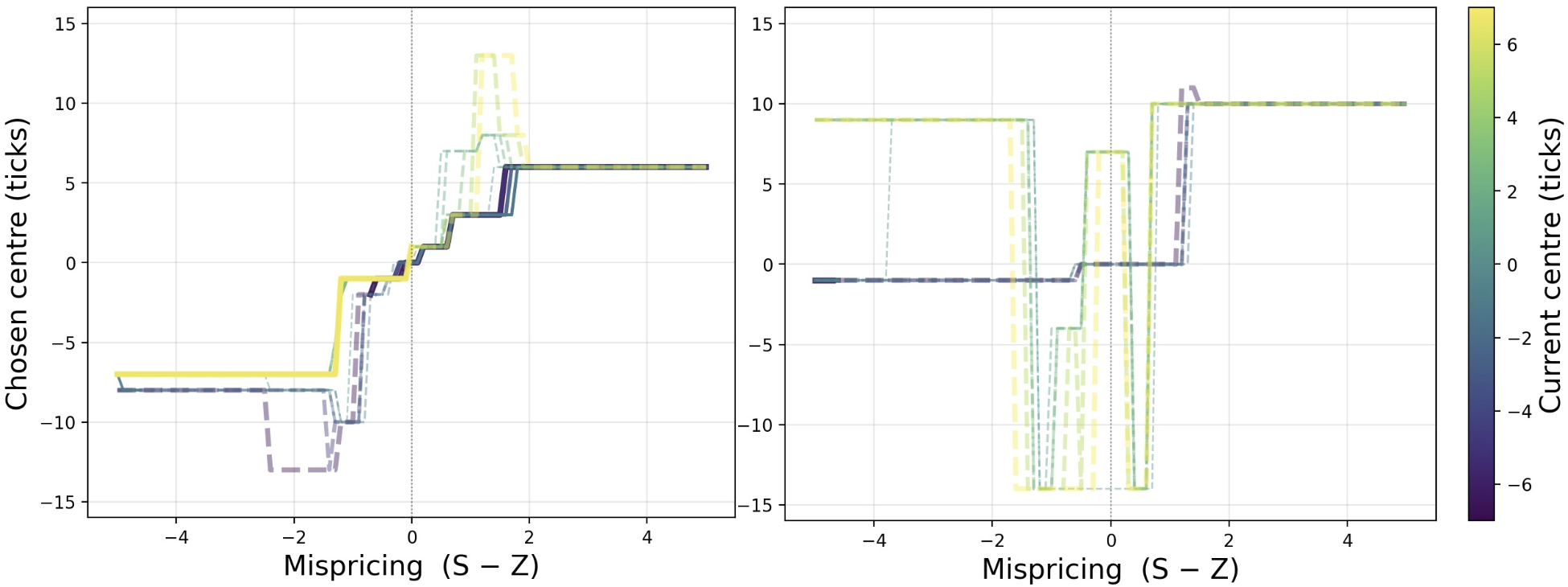}
\label{fig:gas2_actions}\vspace{-0.2cm}
\caption{Action responses of the risk-neutral PPO\_narrow agent for $\sigma=0.01$ with gas cost $g=2$ (left panel) and $g=6$ (right panel) at time $t=0.4$. The x-axis shows the mispricing (USDC), while the y-axis shows the policy output, i.e. the chosen range centre. The colour bar indicates the current centre (ticks). Both the chosen action and current state are expressed as centre offsets. Dark purple denotes a position well below the current tick, bright yellow well above it, and teal positions partially in range. Solid line segments indicate rebalancing; dashed segments indicate that the agent ignores its centre output and holds its position.}
\label{fig:ppo_narrow_actions}
\end{figure}

\noindent Figure \ref{fig:triangles} illustrates this mechanism by superimposing the average positions of PPO agents trained with different risk profiles and a wide passive existing background liquidity component. The resulting aggregate allocation forms a bell-shaped liquidity profile around the current pool price, providing a model-based illustration of the real
\noindent liquidity concentration patterns observed in DeFi pools.
Flatter empirical distributions, such as the one observed on Ethereum, can therefore be associated with stronger risk aversion, higher gas costs or heterogeneous individual LP objectives. Analogously, the liquidity distribution of Optimism and Arbitrum can be thought qualitatively as the superposition of LPs that are more risk seeking in their approach of capturing the fee share of incoming trades by narrowing their quotes, resembling our theory informed RL agent. In this sense, the shape of liquidity in concentrated AMMs reflects both the local fee incentives of CL and the idiosyncratic risk-management behaviours of the LP population. 

\textbf{Profiling Learned Policies as State-Dependent Functions:} \Cref{fig:ppo_ra_actions} and \ref{fig:ppo_narrow_actions} examine the policies as state functions to explain how the agents condition their actions on the observed state. Figure \ref{fig:ppo_narrow_actions} shows the response of the optimal risk-neutral PPO\_narrow agent. Since this agent has a fixed two-tick width, the relevant decisions are where to centre the position and whether or not to rebalance. In the low gas-cost case (left subplot), the agent rebalances positions that are far out of range and places liquidity in the direction of expected arbitrage flow. Specifically, when mispricing is relatively small (-1,1) the positions that are deeply out of range (bold yellow and purple) get rebalanced near the current centre, adjusted to the direction of the arbitrage. However, when the current position is already favourably placed, the agent often chooses to hold, since rebalancing would require paying an additional gas cost. When mispricing is large, the out of range positions get rebalanced and their liquidity is placed again in the direction of the arbitrage however, with heavier asymmetry towards the external price to achieve better rates while collecting the fees. The positions that are already in the direction of the arbitrage (e.g. purple lines with $S-Z<0$ and, conversely, yellow lines with $S-Z>0$) are not rebalanced. The existing placement is expected to yield higher fees collected than rebalancing closer to $S$, due to incurred gas costs.

The high gas-cost case (right subplot) demonstrates that this behaviour changes sharply when intervention becomes expensive. All lines are dashed, which means that the agent holds her current position (even with large mispricing signal) and the chosen centre output can be ignored. Thus, high gas costs not only reduce PnL mechanically but also expand the inaction region. The cost of rebalancing to correct her position will be larger than the expected profit, in contrast to the left subplot, where the gas cost was lower.

Figure \ref{fig:ppo_ra_actions} reports the corresponding state-function analysis for a risk-averse PPO agent. PPO controls with her actions $Z^\ell_n,Z^u_n$ essentially both where (i.e. the centre defined as $Z^c_n:=\frac{Z^\ell_n+Z^u_n}{2}$) and how wide (i.e. a symmetrical half-width $\frac{|Z^u_n-Z^\ell_n|}{2}$ around $Z^c$) to place her liquidity. The figure shows how the policy responds jointly to mispricing and inventory holdings. Specifically, \Cref{subfig:center_ppo_ra} shows that when $S-Z<0$, the agent repositions the centre of her liquidity towards the external price, independently of the amount of ETH held. In doing so, she collects fees from the arbitrage flow, while avoiding exchanging her liquidity in the worst possible rate (the current one) due to the expected incoming arbitrage trades. However, as shown in \Cref{subfig:width_ppo_ra} with $S-Z<0$, the chosen width varies with the amount of ETH held. For low ETH (i.e. more USDC) holdings, the agent widens her position reflecting her risk management. She dilutes USDC capital across the arbitrage direction, reducing exposure to unfavourable prices. For higher ETH holdings (less USDC), she identifies the profitable rates, resulting in a more concentrated but still wide position. Lastly, \Cref{subfig:hold_ppo_ra} shows that for $S-Z>0$ the agent chooses to hold her asset at (-1,1) around the current tick and wait for trades towards the direction of the arbitrage. As a risk-averse agent, she prefers to take the trade because she will exchange the risky token for the safe token (numéraire). 

Overall, Figures \ref{fig:ppo_ra_actions} and \ref{fig:ppo_narrow_actions} show that the learned policies are economically interpretable: they use mispricing to infer order-flow direction, gas costs to decide whether intervention is worthwhile, and inventory holdings to control risk.

\section{Limitations}

Our analysis is deliberately scoped to a controlled, model-based setting. Each RL policy is evaluated in the same stationary economic regime in which it is trained, allowing us to isolate and interpret learned economic mechanisms without confounding them with regime-detection or adaptation errors. The results should therefore be understood as conditional on correctly identified, approximately stationary transition dynamics. The benchmarks are constrained by the concentrated-liquidity literature. We compare the RL agents with tuned static and rebalancing heuristics and with the principal analytical strategy available for CPM liquidity provision. We do not claim universal superiority of RL; rather, we show that it can produce interpretable, state-dependent policies where closed-form solutions are generally unavailable. 
These policies may serve as reproducible model-based references for future work. 
The continuous-time impulse-control problem is implemented with discrete decision opportunities, reflecting the discrete execution of on-chain transactions. The reported policies, therefore, solve the discretized problem induced by the chosen decision frequency, not the exact continuous-time formulation. Finally, the superposition experiment shows that heterogeneous risk preferences can resemble liquidity patterns like those in concentrated-liquidity pools, without identifying them as a unique explanation. An empirical approach would require on-chain data analysis, which is beyond the scope of this work.

\section{Conclusion}
\label{sec:conclusion}
%%%%%%%%%%%%%%%%%%%%%%%%%%%%%%%%%%%%%%%%%%%%%%%%%%%%%%

We studied dynamic liquidity provision in AMMs with CL through an RL perspective. We formulated the LP’s problem as a stochastic impulse control problem, where the LP jointly decides when to rebalance and where to allocate her capital under price uncertainty, arbitrage trades, and operational frictions while managing inventory risk. Due to the analytical intractability of the problem we employed RL algorithms to obtain interpretable state-dependent policies. The results show that learned policies recover economically meaningful behaviour with complex state-dependent strategies. Higher volatility and gas costs reduce average PnL, while stronger risk aversion compresses the left tail of the PnL distribution by inducing wider positions and lower rebalancing rates. Overall, the learned policies adapt their position centre, width, and intervention decision to the state of the market. Lastly, by superimposing learned policies with heterogeneous risk profiles, our modelling provides one possible mechanism capable of generating qualitatively similar concentration patterns with bell-shaped liquidity distributions around the pool prices that are observed in real DeFi exchanges. Future work could extend to multi-agent competition among heterogeneous LPs \cite{baggiani2026competition}, and AMM designs with dynamic trading fees \cite{baggiani2025optimal}.

\bibliographystyle{ACM-Reference-Format}
\bibliography{bibliography}
\end{document}